\documentclass[prx,twocolumn,tightenlines,aps,epsf,showpacs,superscriptaddress,longbibliography]{revtex4-1}
\usepackage[pdftex]{graphicx}

\usepackage{dcolumn}
\usepackage{bm}
\usepackage{epsfig}
\usepackage{latexsym} 
\usepackage{amsmath}
\usepackage{amssymb}
\usepackage{color}
\usepackage{array}
\usepackage{bbm}
\usepackage{color}
\usepackage{array}
\usepackage{cancel}
\usepackage[dvipsnames]{xcolor}
\usepackage[normalem]{ulem}
\usepackage[mathscr]{euscript}
\usepackage{grffile}
\usepackage[colorlinks=true,allcolors=blue]{hyperref}%

\usepackage{graphicx}
\usepackage{physics}
\usepackage{bbm}
\usepackage{amsmath}
\usepackage{amsthm}
\usepackage{amssymb}

\usepackage{tikz}
\usepackage{hyperref}  
\hypersetup{colorlinks=true,citecolor=blue,linkcolor=blue} 

\usepackage[nameinlink,capitalize]{cleveref}
\usepackage{textcomp}
\usepackage{multirow}

\newtheorem{theorem}{Theorem}[section]
\newtheorem{lemma}[theorem]{Lemma}
\newtheorem{definition}[theorem]{Definition}

\newtheorem{problem}[theorem]{Problem}

\begin{document}
\title{Straddling-gates problem in multipartite quantum systems}
\author{Yuxuan Zhang }
\email{yuxuanzhang@utexas.edu}
\affiliation{Center for Complex Quantum Systems, The University of Texas at Austin, Austin, TX 78712, USA}
\date{\today}
\begin{abstract}
We study a variant of quantum circuit complexity, the binding complexity: Consider a $n$-qubit system divided into two sets of $k_1$, $k_2$ qubits each ($k_1\leq k_2$) and gates within each set are free; what is the least cost of two-qubit gates ``straddling'' the sets for preparing an arbitrary quantum state, assuming no ancilla qubits allowed? Firstly, our work suggests that, without making assumptions on the entanglement spectrum, $\Theta(2^{k_1})$ straddling gates always suffice. We then prove any $\text{U}(2^n)$ unitary synthesis can be accomplished with $\Theta(4^{k_1})$ straddling gates. Furthermore, we extend our 
results to multipartite systems, and show that any $m$-partite Schmidt decomposable state has binding complexity linear in $m$, which hints its multi-separable property. This result not only resolves an open problem posed by Vijay Balasubramanian, who was initially motivated by the ``Complexity=Volume” conjecture in quantum gravity, but also offers realistic applications in distributed quantum computation in the near future.

\end{abstract}
\maketitle
\section{Introduction}
    For decades, quantum circuit complexity has been an area of great interest for researchers across disciplines. Given a certain task, such as preparing a pure quantum state from the all-zero state, quantum circuit complexity addresses the cost of implementation (whether exactly or approximately) in terms of a set of minimum operations, which usually consists of single-qubit and two-qubit gates. However, depending on the scenario, not all quantum gates are created equally: in a near-stabilizer setting, for example, $T$ gates are weighted higher as a resource cost than Pauli operations, including the CNOT gates, since the simulation cost is polynomial in the number of stabilizers but exponential in that of the non-stabilizer ones \cite{bravyi2016improved}. Conversely, in a quantum hardware implementation, experimentalists tend to view CNOT gates as more expensive than the $T$ gates due to their higher noise rate \cite{arute2019quantum,pino2021demonstration,egan2020fault}. In both examples, one aims to minimize the use of higher cost operations.
    
    Despite its original purpose to study the interior volume of multiboundry wormholes \cite{Balasubramanian:2018hsu}, binding complexity can be best motivated from an experimentalist's perspective. One key obstacle of building a ``useful'' quantum computer is that current quantum hardware platforms, whether superconducting or trapped-ion based, are hard to scale~\cite{Preskill2018quantumcomputingin}. For instance, regardless of their high entangling gate fidelity \cite{pino2021demonstration,egan2020fault}, the cutting edge trapped-ion computers might soon hit a limit in qubit number, should they stick with the current $1d$ geometry \cite{bruzewicz2019trapped}. A natural solution is to have two or more smaller quantum computers working collaboratively \cite{beals2013efficient}, which leads to another concern: the unitary operations ``straddling'' between computers are unreliable. This is especially vital if one wants to perform physically long-range quantum operations or transduction between different physical platforms (for instance, from a microwave-controlled computation system to a photon based information carrier) \cite{lauk2020perspectives}. Overall, in an idealized model, we are only interested in the minimum number of total straddling gate cost to accomplish a task, and all local operations shall be considered free. 
    
    Following circuit complexity, the binding complexity, $\mathcal{C}_b$, is defined as the minimum amount of straddling two-qubit operations required to accomplish a certain given task \cite{Balasubramanian:2018hsu}. This concept quantifies the difficulty of distributing entanglement among multiple parties, and the original authors carried out detailed calculations on instances of physicists' interests. Nonetheless, the binding complexity for synthesizing an $arbitrary$ quantum state or unitary, which is natural to ask about from a quantum information perspective, has remained open \footnote{A version of the problem came up during a private conversation between Vijay Balasubramanian and Scott Aaronson, and the latter suggested a counting argument for estimating the lower bound}. 
    
    In this work, we first study the binding complexity for circuit synthesis in any bipartite quantum systems. We then extend the result to multipartite systems: there, through a case study of Schmidt decomposable states, we reveal binding complexity's potential in measuring entanglement. Lastly, we extend our discussion to future open problems. Unless otherwise specified, we take the all-zero state as the only reference state and restrict the straddling gate set to be any SU(4) gate, which can be decomposed with CNOT gates and single qubit rotations with a ``merely'' constant overhead \cite{10.5555/968879.969163}.

\section{Bipartite systems}\label{section:bs}
Quantum circuit synthesis is studied in many quantum computing subjects such as compiler designs \cite{khatri2019quantum}. For instance, given a desired quantum circuit, in order to reduce noise rate and to save computation time, a common practice is to figure out a minimum decomposition in the ``natural'' gate set of a quantum hardware. Here, we study the question in minimizing only the cost of straddling two qubit gates between two parties. This means any unitary operations within each party will be considered free.
\subsection{State preparation}
To begin with, let us consider the following task: Suppose Alice and Bob share $n$ qubits and they want to prepare some given quantum state, how can they minimize the number of straddling operations? Here, we give an asymptotically optimal result.
\begin{problem}[Bipartite state preparation (BSP) problem]
Consider a system with n qubits, divided into two halves of $k_1$, $k_2$ qubits each, where $k_1 \leq k_2$. The goal is to prepare some given pure state $\ket{\psi}$ on those qubits, via a circuit consisting of two-qubit gates. Gates within each half are free; accomplish the task with as few straddling gates between the halves as possible without using ancilla qubits.
\end{problem}
To gain some insight into this problem, observe that for any $\ket{\psi}$ defined on $n$ qubits and divided into region $A$ and $B$, there exists a Schmidt decomposition \cite{Orus:2013kga}:
\begin{equation}\label{sch-state}
    \ket{\psi} =\sum_{i} w_{i}\ket{\psi_{A}^{i}}\otimes\ket{\psi_{B}^{i}}
\end{equation}
where $w_i$'s denote Schemidt weights with normalization condition $\sum_i \abs{w_i^2} = 1$; additionally, $\{\ket{\psi_{A}^{i}}\}$ and $\{\ket{\psi_{B}^{i}}\}$ should both form orthonormal basis in $\mathcal{H}_{A}$ and $\mathcal{H}_{B}$. This requires the Schmidt rank 
$$r_S\leq \operatorname{min}(\abs{\mathcal{H}_{A}},\abs{\mathcal{H}_{B}});$$
i.e., in a $k_1$, $k_2$ bipartite system, there exists at most $2^{k_1}$ Schmidt weights. Intuitively, the entanglement space is practically homomorphic to a $k_1$-qubit system, with $r_S$ being the entanglement degree of freedom. Any local unitary operation within each subsystem alone would not contribute to the entanglement spectrum, whereas the operator Schmidt decomposition of any straddling $g\in \text{SU}(4)$ contains only a constant number of free parameters \cite{Balasubramanian:2018hsu}.
By this counting argument, $\Omega(r_S)$ gates are required to generate a dense set for state preparation. It turns out there exists a matching upper bound for $\mathcal{C}_b$:
\begin{theorem}[A solution to BSP problem]\label{QD}
$\mathcal{C}_b(\ket{\psi}) = \Omega(r_S)$ for any bipartite state $\ket{\psi}$ with Schmidt rank $r_S$.
\end{theorem}
We give an explicit algorithm to accomplish the task:
\begin{enumerate}
\item Assume computation basis $\{\ket{A^i}\}$ and $\{\ket{B^i}\}$, namely $\{\ket{0}^{\otimes k_1},...\ket{1}^{\otimes k_1}\}$ and $\{\ket{0}^{\otimes k_2},...\ket{1}^{k_2}\}$. Prepare the state $\sum_{i} w_{i}\ket{{A}^{i}}$ in region $A$, so that we obtain $\sum_{i} w_{i}\ket{A^{i}}\otimes\ket{0^{\otimes{k_2}}}$. 
\item Use straddling gates between $A$ and $B$ to construct $\sum_{i} w_{i}\ket{A^{i}}\otimes\ket{B^{i}}$. 
\item Apply local transformations $\sum_{j}\ket{\psi_{A}^{i}} \bra{{A}^{j}}$ and $\sum_{j}\ket{\psi_{B}^{i}} \bra{{B}^{j}}$ to each subsystems, and get back the desired state.
\end{enumerate}
\begin{figure}
    \centering
    \includegraphics[width=0.4\textwidth]{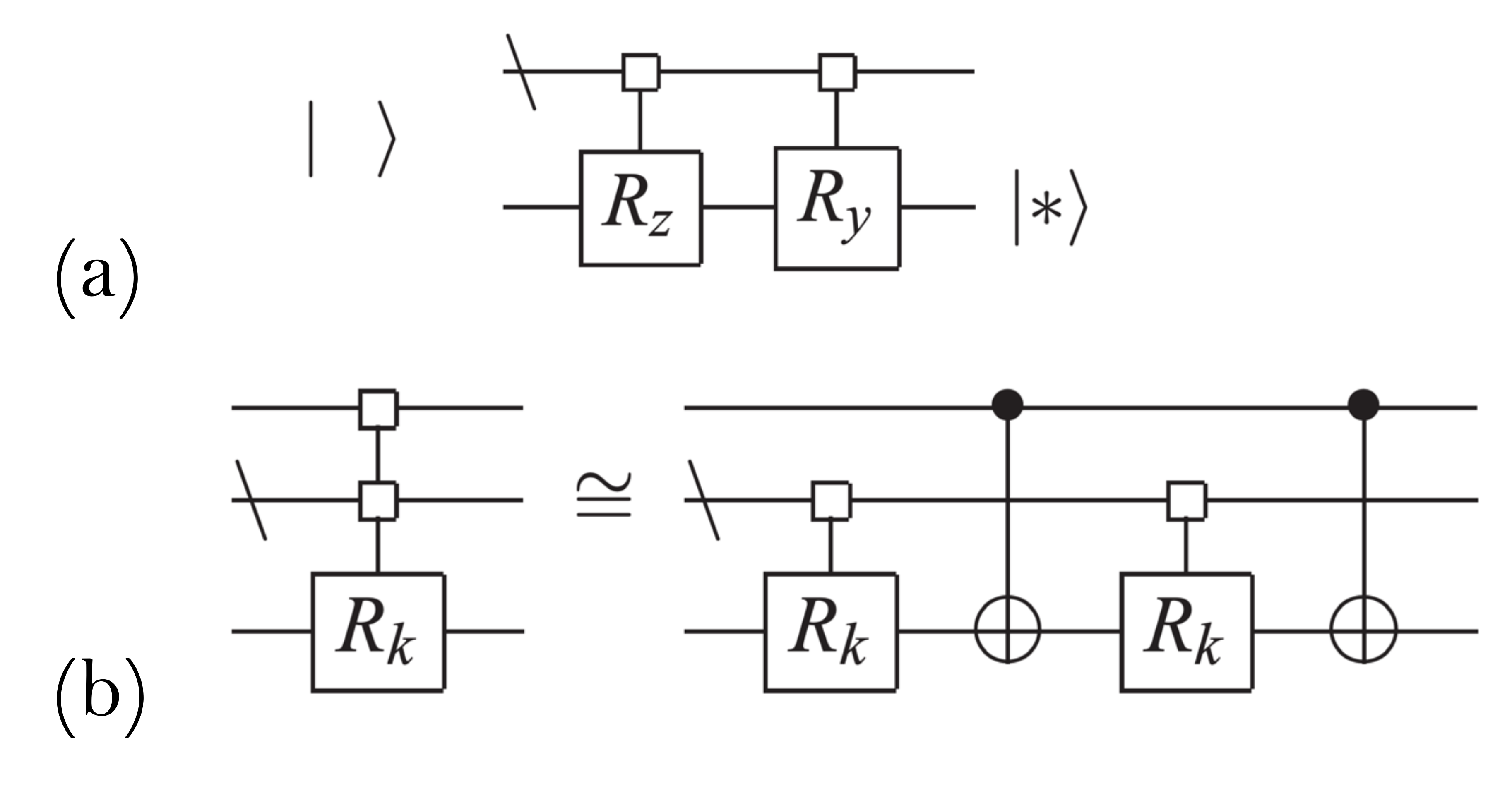}
    \caption{{\bf Disentangling a qubit} Top: For any input quantum state, one could disentangle any bit with two multiplexors; reversing this process gives an iterative algorithm for state preparation~\cite{10.5555/968879.969163,PhysRevLett.93.130502}. Bottom: decomposing a $c$-control multiplexor with two $(c-1)$-control multiplexors and two CNOT gates. In the diagram, the register with slash is an abbreviation for all control registers, white box denotes control bits and $R_k$ is a rotation gate along axis $k\in{y,z}$ with angles to be determined.}
    \label{DE}
\end{figure}
No straddling gate is used in step 1 and 3, and, crucially, the number of qubits we need to entangle in step 2 is only directly related to $r_S$. In other words, the Schmidt spectrum can be encoded with $\lceil\text{log}(r_S)\rceil$ qubits, which is independent of system size. The second step requires $O(r_S)$ straddling gates, a result from the following fact:
\begin{lemma}[Bipartite qubit entangling]\label{bqe}
Entangling two systems with $k_1$, $k_2$ qubits respectively (assuming $k_1\leq k_2$) requires no more than $O(2^{k_1})$ straddling \text{CNOT} gates.
\end{lemma}
 Next, we make use of known quantum circuit synthesis techniques to prove Lemma~\ref{bqe}.
\begin{proof}
As shown in Fig.\ref{DE}, one can prepare a quantum state iteratively with multiplexors, a special type of quantum gates. A multiplexor can be thought as a generalization of controlled gates with multiple control qubits, which can be synthesized with a recursive construction. Notice that, not all control bits for the gates are within the same subsystem; consequently, one has to specify the locations of control and target bits to conclude their straddling costs. Assuming a non-trivial case where control bits can belong to two regions: let $p$ be the number of control bits in region $A$ and $q$ be that in $B$. Without loss of generality, also assume the target bit is in region $B$. 
We would like to derive a decomposition cost for multiplexors, denoted as $T(p,q)$; the construction in Theorem~\ref{DE} implies: 
\begin{align}\label{D1}
    T(p,q) &= 2T(p,q-1)
\end{align}
or
\begin{align}\label{D2}
    T(p,q) &= 2T(p-1,q) + 2 
\end{align}
 apparently $T(0,q) = 0$; the best known straddling cost, to the author's knowledge is $T(p,0) = (2p - 1)$~\cite{PhysRevA.52.3457}. Solving Eq.~\ref{D1} delivers $T(p,q) = (2p - 1)2^{q}$, whereas Eq.~\ref{D2} gives $T(p,q) = 2^{p+1}-2$. Following Shende et al., one could bring $T(p,q)$ further down to $2^p$ due to cancellations between the CNOT gates~\cite{1629135}. This suggests that, when a $p\leq q$, the best strategy is to perform reduction in the control space. 
 
 Next, we already know that entangling each qubit requires two multiplexors; therefore, the binding complexity for preparing any state in the $k_1 + k_2$ system in Lemma~\ref{bqe} is: 
 \begin{align}
     \mathcal{C}_b(\ket{\psi}) &\leq 2[T(k_1-1,k_2) + T(k_1-2,k_2) + ...]\\
     &\leq 2^{k_1}-2
 \end{align}
which completes the proof. 
\end{proof}

Putting everything together, we showed $\mathcal{C}_b(\ket{\psi}) = \Theta(r_S)$. The result is in fact rather satisfying: the complexity of combining two quantum system is linear in the entanglement degree of freedom, which can be especially useful in preparing low bipartite entanglement states. 
In fact, Lemma \ref{bqe} can be thought as a looser bound of Theorem \ref{QD}, without making any assumption on a state's entanglement spectrum. In a full Schmidt rank case, our result suggests $\mathcal{C}_b(\ket{\psi}) =\Theta(2^{k_1})$. For a small $k_1$ that is independent of $n$, the straddling cost does not depend on $n$; on the other hand, the binding complexity maximizes at $k_1 = n/2$, where $\mathcal{C}_b(\ket{\psi}) =\Theta(2^{n/2})$.

\subsection{Unitary decomposition}
 One could further question about the bipartite binding complexity for applying a given unitary transformation, $\mathcal{C}_b(U(2^n))$. In this section, we prove the tight answer is $\Theta(4^{k_1})$. 

\begin{theorem}[Bipartite unitary decomposition]
Any given $n$-qubit unitary can be constructed with $O(4^{k_1})$ straddling gates in an equal bipartite system.
\end{theorem}

\begin{proof}
As shown in Fig. \ref{QSD}, an arbitrary $n$-qubit quantum operator can be implemented recursively by a circuit containing three controlled rotations and four generic $(n-1)$-qubit operators \cite{1629135}. We adopt their approach and solve the recurrence for the bipartite binding complexity. The straddling cost, $\mathcal{C}_b(p,q)$, of a $(p+q)$-qubit unitary has the relationship:
\begin{figure}
    \centering
    \includegraphics[width=0.48\textwidth]{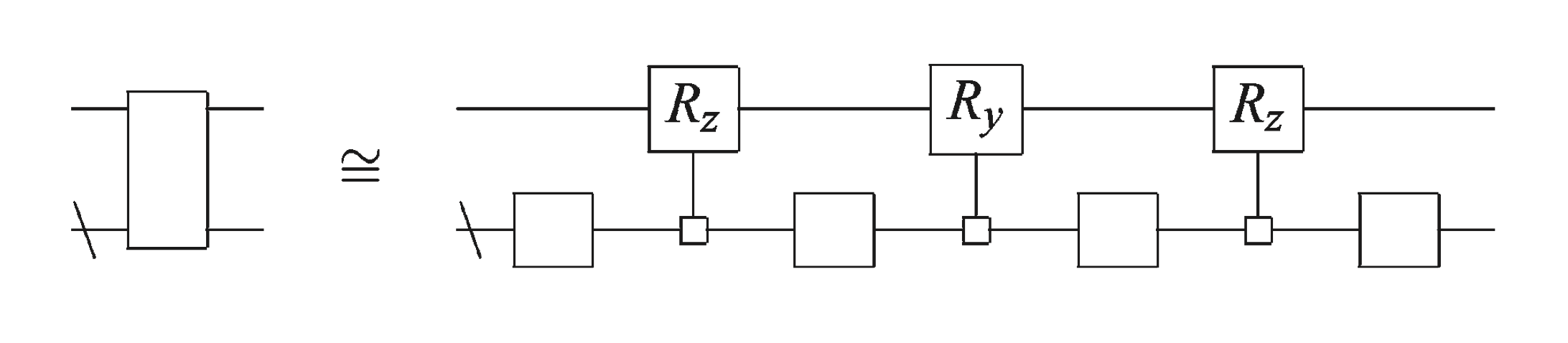}
    \caption{{\bf QSD process for an arbitrary unitary} The diagram shows a recursive construction of a arbitrary unitary $U$, also know as the ``quantum Shannon decomposition (QSD)'' \cite{1629135}.}
    \label{QSD}
\end{figure}
\begin{align}\label{reQSD}
    \mathcal{C}_b(p,q) \leq 4\mathcal{C}_b(p-1,q) + 3T(p-1,q)
\end{align}
We have previously obtained a upper bound for decomposing any mutiplexor, $T(p,q)$. By induction, one can prove the following inequality for $C$:
\begin{align}\label{inQSD}
    \mathcal{C}_b(p,q) \leq & 4^{p-l}(\mathcal{C}_b(l,q) + 3\cross2^{p-1})\\
                &-3\cross2^{p}
\end{align}
$C(0,q)=0$ is the end of the induction since the unitary would be fully in region $A$. For an euqal bipartite system, inserting values $p=k_1$ and $q = k_2$ concludes $\mathcal{C}_b(U(2^n)) = O(4^{k_1})$.

\end{proof}
The lower bound for unitary decomposition can be shown with a counting argument as well: the SU($2^n$) has $O(4^n)$ parameters; between every straddling gate, one can perform a pair of local unitary operations, which contains $O(4^{k_1}+4^{k_2}) = O(4^{k_2})$ parameters. Therefore $\Omega(4^{k_1})$ straddling gates are required. 

\section{Multipartite systems}\label{section:ms}
In this section we extend our discussion to multipartite systems. Foremost, we extend our binding complexity results to for general state preparation and unitary decomposition; secondly, we discuss binding complexity as an indicator for robustness of entanglement, which can be used to distinguish equivalence classes under stochastic local operations and classical communication (SLOCC); lastly, we perform a case study on the binding complexity's Schmidt decomposable states and comment on their multi-separability.
\subsection{Overview}
Let's consider a $m$-partite system with $k_1,k_2,...k_m$ qubits on each party, where $k_1+k_2+...+k_m = n$ and $k_1\leq k_2\leq...\leq k_m$. Then for any constant $m$, the following fact holds:

\begin{theorem}[Multipartite state preparation] \label{lb}
Arbitrary quantum state preparation in the system requires $\Theta(2^{n-k_m})$ straddling gates.
\end{theorem}
To quickly validate this, notice that when $m\rightarrow{n}$, the bound asymptotically approaches $\Omega(2^n)$, the typical two-qubit gate complexity for an $n$-qubit state; on the other hand, $m=2$ recovers the equal bipartite result $\Theta(2^{k_1})$. Rigours proofs can be established with tools introduced in Sec.~\ref{section:bs}, and a sketch of the idea is given below.

For the lower bounds, the same counting argument works nicely: one could perform Schmidt decomposition between each party with a increased order in their qubit number:
\begin{align}\label{mul-sch-state}
    \ket{\psi} &=\sum_{i_1} w_{i_1}\ket{\psi_{1}^{i_1}}\otimes\ket{\psi_{2,..m}^{i_1}}\\ 
    &=\sum_{i_1,i_2} w_{i_1}w_{i_1,i_2}\ket{\psi_{1}^{i_1}}\otimes\ket{\psi_{2}^{i_1,i_2}} \otimes\ket{\psi_{3,..m}^{i_1,i_2}}\\ 
     &=\sum_{i_1,..i_{m-1}} w_{i_1}..w_{i_1,..i_{m-1}}\ket{\psi_{1}^{i_1}}\otimes..\otimes\ket{\psi_{m}^{i_1,..i_{m-1}}}
\end{align}
The total degree of ``straddling'' freedom in this summation would be of order $\Omega(2^{n-k_m})$, even after considering orthogonality constraints. Therefore, $\Omega(2^{n-k_m})$ straddling gates are required. For the upper bound, observe that disentangling all other qubits from the $k_m$ qubits can, once again, be done with multiplexors, and the total two qubit straddling gate count in doing so goes asymptotically $\sim 2^{n-k_m}$, which completes the proof.

Notice that, unlike in the bipartite state preparation case, we have not made assumption on the given state's Schmidt spectrum, due to the complication in describing entanglement in a multipartite system. Nevertheless, it would be of future interest to connect binding complexity with so-called matrix product state (MPS) ~\cite{perez2006matrix} representation, since it is widely used to compress low entanglement state in condensed matter physics and indirectly encoded quantum simulations~\cite{foss2021holographic,niu2021holographic,zhang2022qubit}. As a natural generalization of Theorem~\ref{QD}, the direct relationship between a state's MPS rank and its circuit complexity should be expected.
\subsection{Robustness of Entanglement}
Multipartite entanglement has been studied in abundant contexts: many-body physics~\cite{perez2006matrix}, quantum information~\cite{bruss2002characterizing}, cryptography~\cite{PhysRevA.60.910}, and even quantum gravity~\cite{Susskind:2014rva}. A great amount of effort has been concerned with characterizing quantum correlation between systems by studying their entropy as well as other generalized, Von Neumann entropy-based measures~\cite{Bagchi_2015, Thapliyal_1999,bruss2002characterizing}. Nevertheless, it was proposed that binding complexity can serve as an alternative approach to quantify entanglement~\cite{Balasubramanian:2018hsu}, and we elaborate on this possibility below. 

Since entropy is often used as the golden measure for entanglement in bipartite systems, we start by demonstrating a toy example, revealing its dissimilar with binding complexity. Consider an entangled state $\ket{\psi_{AB}} = \sqrt{1-\epsilon}\ket{00}+\sqrt{\epsilon}\ket{11}$, the entanglement entropy for qubit $A$ is 
$$S_A = -\text{Tr}(\rho_A \text{log} \rho_A) = (\epsilon-1)\text{log}(1-\epsilon)-\epsilon \text{log}{\epsilon} $$
which is $\epsilon$ dependent; yet the state consists $\mathcal{C}_b = 1$ no matter what the value of $\epsilon$ is. Binding complexity is only concerned with the existence of entanglement in $\ket{\psi}$, whereas entropy also cares how close $\ket{\psi}$ is to a separable state. Shall we compare quantum information as water flowing between regions in a piping system, entropy measures the magnitude of water across a subsystem, whereas binding complexity only cares about the maximum capacity of pipes connecting sub-regions. 

Consequentially, binding complexity tend to be a better probe if the entanglement \emph{structure} is of interest. In the $m=n$ setting (each party is a qubit), the Greenberger-Horne-Zeilinger (GHZ) states are famously viewed as a different entanglement class from the W states (for these states, each subsystem is a qubit):

\begin{align}\label{ghz_w}
    &|\psi_{GHZ}\rangle = \dfrac{1}{\sqrt{2}} (|0\cdot \cdot \cdot 00\rangle +|1\cdot \cdot \cdot 11\rangle )\\
    &|\psi_{W}\rangle = \dfrac{1}{\sqrt{n}} ( |0\cdot \cdot \cdot 01\rangle + |0\cdot \cdot \cdot 10\rangle + \cdot \cdot \cdot + |1\cdot \cdot \cdot 00\rangle)
\end{align}

In the special case of $n=3$, these states can be shown to represent two distinct equivalence classes by demanding state conversion through stochastic local operations and classical communication (SLOCC) \cite{PhysRevA.62.062314}. This inequivalence is reflected on their physical property: the GHZ state is multi-separable; that is, its reduced density matrix becomes separable upon partial tracing any of its parties. Intuitively, a separable density matrix represents a classical ensemble of product states that can be prepared with local transformations and classical communication (LOCC) \cite{Thapliyal_1999}. This means no quantum coherence stands between parties after the partial trace. On the other hand, a W state maintains two body entanglement even after ``losing'' any of its parties. However, calculating the bipartite entanglement entropy suggests $S_{\text{GHZ}}>S_{\text{W}}$, which means entropy fails to capture a key characteristic in this example: the robustness of entanglement.

Conversely, a W state after any party being partial traced cannot be written into such form, suggesting quantum entanglement is still shared within at least some surviving parties. So, what does binding complexity say about these states? Here, since we want to be precise about constants, we focus on the counting of CNOT gates. It was previously proved that \cite{vznidarivc2008optimal}:

\begin{lemma}
In a three qubit system, $\mathcal{C}_b(\ket{\psi_{\text{W}}}) = 3$ and $\mathcal{C}_b(\ket{\psi_{\text{GHZ}}}) = 2$.
\end{lemma}

This can be proved in the following logic: two CNOT gates are necessary to entangle three qubits, and turns out to be also sufficient for GHZ-like states; however, it can be straightforwardly shown that any local unitary operation cannot transfer a GHZ-like state into W. In fact, in the m-partite case, a GHZ state has binding complexity $\Theta(m)$ whereas $\mathcal{C}_b$ for a W state is conjectured to be $\Theta(m^2)$ \cite{Balasubramanian:2018hsu}. The lower binding complexity provides an explanation for GHZ states' lack of robustness against local measurements. Furthermore, we may even take one more step and generalize the GHZ result to a class of so-called \emph{Schmidt decomposable} states:

\begin{definition}We say an $m$-partite quantum state is a Schmidt decomposable state if:
\begin{align}\label{deco}
    \ket{\psi} =\sum_{i} w_{i}\ket{\psi_{1}^{i}}\otimes\ket{\psi_{2}^{i}}\otimes\cdot \cdot \cdot\ket{\psi_{m}^{i}} 
\end{align}

where $\{\ket{\psi_{1}^{i}}\}$, $\{\ket{\psi_{2}^{i}}\}$, $\cdot \cdot \cdot$ $\{\ket{\psi_{m}^{i}}\}$ form orthonormal bases in their own subspaces.
\end{definition}


Schmidt decomposable states can be thought of as a more generic GHZ state where each party is allowed to possess multiple qubits, while all parties still share the same Schemidt spectrum. Like GHZ, Schmidt decomposable states' reduced density matrices become separable upon tracing out any subset of the parties. This property is, likewise, reflected in its $\mathcal{C}_b$'s linear dependence on $m$.

\begin{theorem}[Schmidt decomposable state preparation]\label{s-d}
A system with $n$ qubits is divided into $m$ parts. Any Schmidt decomposable state in a such system can be prepared with $O(mr_S)$ straddling gates, where $r_S$ denotes the state's Schmidt rank.
\end{theorem}

We've seen that, for a generic state, the number of Schmidt weights grows as one performing Schmidt decomposition; it, however, does not grow with $m$ for a Schmidt decomposable state, preserving the one-to-one correspondence between any two local base vectors across sub-regions. Thus we could encode Schmidt weights on any of the subsystems and entangle all remaining subsystems to it. Specifically, one could simply extend the algorithm in Theorem \ref{QD} to all $m$ subsystems: in step 2, no multiplexor's target bit is within the encoded space; therefore, the Schmidt weights are untouched and transmitted to every party at a extra linear cost in $m$. To sum up, the linear growth in binding complexity suggests a state's multi-separability.


\section{Discussion} 
After inspecting several instances, perhaps the next emerging open question one could ask about binding complexity is whether it serves as a good measure for entanglement. There are four necessary conditions a measure should satisfy \cite{PhysRevLett.78.2275}:
\begin{enumerate}
    \item Zero for separable states.
    \item Invariant under local unitary transformations.
    \item Additive for tensor products.
    \item Non-increasing under classically coordinated local operations.
\end{enumerate}
For pure states, the first three conditions are evidently true by definition of binding complexity. Though the usage of and classical communication is not within the main scope of this work, it would be counter-intuitive to assume classical manipulations would increase the number of entanglement gate required to prepare a state. On the other hand, if binding complexity does increase with LOCC, it would be interesting to understand how that happens. Another crucial evidence is, in bipartite case $\mathcal{C}_b \sim r_S$, and the latter was shown to be non-increasing under LOCC \cite{terhal2000schmidt}. We therefore conjecture that binding complexity can be used as a measure to quantify entanglement, at least for pure states. Surely, a careful examination of binding complexity classes under the use of LOCC is worth exploring.

Beyond an attempt to solve an open problem, our work could provide inspirations in various ways: Firstly, looping back to the original motivation, binding complexity provides a tool to study quantum gravity. In the anti-de Sitter/conformal field theory correspondence, a widely-studied conjecture is that a black hole's interior can be mapped to a dual conformal quantum state, whose circuit complexity correspond to the black hole's interior volume, which will grow over time. However, if we consider the \emph{binding} complexity of a multiboundary wormhole, it corresponds to a invariant under local unitary evolution which might be used to characterize wormholes' properties.   

For example, consider the gravity side of the GHZ state, how can ``neglecting'' one party destroy the entanglement of the whole system? Leonard Susskind explains this property by hypothesising the existence of a non-classical ``GHZ-Brane'' localized in the wormhole that connects parties, and it only allows one to receive message sent by the union of all rest parties \cite{Susskind:2016jjb}. This is a local unitary-invariant property that is reflected in the GHZ state's linear binding complexity dependence with the number of parties. Consequentially, it is natural to conjecture the existence of ``GHZ-Brane'' between any Schmidt decomposable state, where the same multi-separable physical property can be observed. 

Moreover, as we discussed in the introduction, binding complexity emerges naturally from the perspective of distributed quantum computation. One might consider more refined versions of the problem, such as allowing the use of $\text{log}(n)$ ancillary qubits or allowing certain rate of quantum communication between the parties. In the near future, we envision that such studies on binding complexity variants could lead to optimizations in experimental designs. 

The author gratefully acknowledges Scott Aaronson for problem suggestion. We thank Vijay Balasubramanian, Arjun Kar, Andrew C. Potter, and Ruizhe Zhang for inspiring discussions. This work was supported by NSF Award No. DMR-2038032.
\bibliography{main.bib}
\end{document}